\documentclass[letterpaper, 10 pt, conference]{ieeeconf}

 \usepackage{graphicx}
\usepackage{color,url}
\usepackage{verbatim}
\usepackage{algorithm}
\usepackage{bm}
\usepackage{algpseudocode}
\usepackage[all]{xy}
\usepackage{url}
\usepackage{amsmath,amsthm,amssymb}
\usepackage{graphicx}
\usepackage{subcaption}

\newtheorem{theorem}{Theorem}[section]

\newtheorem{lemma}{lemma}
\theoremstyle{definition}
\newtheorem{definition}{Definition}

\theoremstyle{remark}
\newtheorem{remark}{Remark} 

\newcommand{\opt}[1]{{#1}^\ast}

\newcommand{\sgn}[1]{\mathrm{sgn}\br{{#1}}}

\DeclareMathOperator*{\maxi}{\text{Maximize}}

\newcommand{\norm}[1]{\left\lVert {#1} \right\rVert}

\newcommand{\tran}[1]{{#1}^T}

\newcommand{\inner}[2]{\left\langle {#1},{#2} \right\rangle}

\newcommand{\tranb}[1]{{\left({#1}\right)}^T}
\newcommand{\br}[1]{\left({#1}\right)}

\newcommand{\inv}[1]{{\left(#1\right)}^{-1}}
\newcommand{\inve}[1]{{#1}^{-1}}

\newcommand{\diag}{\mathrm{di}}
\newcommand{\diagb}[1]{\mathrm{di}\left({#1}\right)}

\newcommand{\R}{\mathbb{R}}

\newcommand{\C}{\mathcal{C}}

\newcommand{\FHinf}[1]{q_\infty}
\newcommand{\FHtwo}[1]{q_2}
\newcommand{\FHone}[1]{q_1}

\newcommand{\E}{\mathcal{E}}

\newcommand{\V}{\mathcal{V}}

\newcommand{\Int}{\mathrm{Int}}

\newcommand{\Jac}[2]{\nabla {#1}\br{#2}}

\newcommand{\pmo}{\psi}

\newcommand{\vall}{\zeta}

\newcommand{\Sym}[1]{\mathrm{Sy}\br{#1}}

\newcommand{\Pa}{W}

\newcommand{\Ma}{A_\alpha}
\newcommand{\Mb}{b}

\usepackage{hyperref}
\IEEEoverridecommandlockouts                              

\graphicspath{ {Figures/} }

\title{\LARGE \bf
Natural Gas Flow Solutions with Guarantees: A Monotone Operator Theory Approach
}

\author{Krishnamurthy Dvijotham$^{1}$, Marc Vuffray$^{2}$, Sidhant Misra$^{2}$ and Michael Chertkov$^{2}$
\thanks{$^{1}$Krishnamurthy Dvijotham is with the Department of Computing and Mathematical Sciences, 
        California Institute of Technology, 1200 E California Blvd, Pasadena, USA
        {\tt\small dvij@cs.washington.edu}}%
\thanks{$^{2}$Marc Vuffra, Sidhant Mishra and Michael Chertkov are with the Center for Nonlinear Studies (CNLS) and T-Division, Los Alamos National Laboratory, Los Alamos, NM}
}

\begin{document}

\maketitle
\thispagestyle{empty}
\pagestyle{empty}

\begin{abstract}
We consider balanced flows in a natural gas transmission network and discuss computationally hard problems such as establishing if solution of the underlying nonlinear gas flow equations exists, if it is unique, and finding the solution. Particular topologies, e.g. trees, are known to be easy to solve based on a variational description of the gas flow equations, but these approaches do not generalize. In this paper, we show that the gas flow problem can be solved efficiently using the tools of monotone operator theory, provided that we look for solution within certain monotonicity domains. We characterize a family of monotonicity domains,  described in terms of Linear Matrix Inequalities (LMI) in the state variables, each containing at most one solution. We also develop an efficient algorithm to choose a particular monotonicity domain, for which the LMI based condition simplifies to a bound on the flows.  Performance of the technique is illustrated on exemplary gas networks. 
\end{abstract}

\section{Introduction}

Natural gas is the fastest growing component of the energy industry. In the USA this growth is primarily due to the hydro-fracking revolution \cite{fracking}, which has provided tremendous increase in supply, driving down prices and stimulating significant network expansion and creating new uses of the natural gas \cite{2010MITEI,2013MITEI}. The revolution also makes operating, controlling and optimizing the emerging gas systems more challenging. Local considerations which prevailed in the past are no longer sufficient to keep operations of the natural gas system economic and reliable. Developing new methods to analyze the state of the growing network efficiently and accurately has now become a top priority for gas system operators and other entities, such as power system operators with many gas fired turbine on balance, interested to monitor the status of the consumption/production, network flows and pressure profile along the long haul gas transmission networks.

Resolving many of the emerging computational challenges in operations and planning of the natural gas networks depends on solving the so called Gas Flow (GF) equations.  These equations describe spatio-temporal distribution of the mass flows and pressures in pipes over an arbitrary network subject to known profile of external parameters characterizing injection/consumption and compression, see e.g. \cite{osiadacz1987simulation,96WS,00WS}. The equations are nonlinear and difficult to analyze, even in the stationary (time-independent) regime discussed in this manuscript,  assuming overall balance of production and consumption,  where thus no accumulation or loss of gas in the system (no, so-called, line pack) takes place. Analyzing this system of the GF equations entails to establishing solution existence, finding the solution and establishing uniqueness, if possible. We are not aware of publications rigorously analyzing existence and uniqueness of solutions to the GF equations. Thus, a natural folklore conjecture would be that in its full generality the problems of existence, uniqueness and finding a solution are difficult,  i.e. the problems are of the complexity exponential in the size of the network. On the other hand, in a couple of special cases the problems were characterized as easy (of the polynomial complexity).  Special cases are known where the GF equations can be solved easily: For tree networks, the GF equations can be solved by dynamic programming \cite{68WL,90LP,13MFBBCP}. Another easy case  is the one of an arbitrary loopy network with no compression, or including only additive (and not multiplicative) compression,  when the problem of finding the balanced gas flow solution becomes equivalent to solving  a convex optimization problem, that of  minimizing cumulative losses in the gas pipes due to turbulent friction subject to flow conservation at any junction of the network \cite{77Nau,12BNV}.

In this manuscript, we apply the tools of monotone operator theory to analyze the GF equations. Just like a convex optimization problem is solvable in polynomial time, zeros of monotone operators can be found in polynomial time \cite{boydmonotone}. Thus, if we can show that the GF equations can be written in terms of a monotone operator, they can be solved efficiently. In this manuscript, we indeed show that GF operator is monotone over restricted \emph{monotonicity domains}. Each monotonicity domain can be expressed as a Linear Matric Inequality (LMI) in the system state (the pressures and flows in the network). The monotonicity domains are parameterized by a matrix-valued parameter. Within each monotonicity domain, there can be at most one solution of the GF equations and determining existence and computation of solutions can be done in polynomial time. We then show that for a particular choice of the monotonicity domain, the LMI condition can be replaced with a simple bound on the ratio of the maximum to minimum flow in the system. Furthermore, the parameters of this monotonicity domain can be computed efficiently by solving a quasi-convex optimization problem.

The rest of this paper is organized as follows. In Section \ref{sec:Background}, we introduce the GF equations and formulate them as a system of multivariate quadratic equations in the flow and potential variables. We also introduce the theory of monotone operators and solution of variational inequalities. In Section \ref{sec:MonGF}, we present our main technical results on the family of monotonicity domains of the gas flow equations. In Section \ref{sec:Num}, we evaluate our approach on some test networks. Finally, in Section \ref{sec:Con} we summarize our contributions and discuss directions for future work.

\section{Background and Notation}\label{sec:Background}

\subsection{Notation}
$\R$ is the set of real numbers. $\R^n$ denote the corresponding Euclidean space in $n$ dimensions.

Given a set $\C\subset\R^n$, $\Int\br{\C}$ denotes the interior of the set. $\norm{x}$ refers to the Euclidean norm of a vector $x\in \R^n$ and $\inner{x}{y}$ to the standard Euclidean dot product. $\diagb{x}$ is the diagonal matrix with with diagonal entries given by $x$. For $x,y\in\R^n$, $x*y$ denotes the vector $z$ with $z_i=x_iy_i$. 

Given a differentiable function $f:\R^k\mapsto\R^k$, $\nabla f$ denotes the Jacobian of $f$, a $k\times k$ matrix with the $i$-th row being the gradient of the $i$-th component of $f$.

For $M\in\R^{n\times n}$, $\Sym{M}=M+\tran{M}$. $\mathcal{S}^k$ denotes the space of $k\times k$ symmetric matrices.

\subsection{Modeling Gas Networks}

The network is specified by a set of buses (nodes) $\V$ and a set of gas pipelines (edges). Nodes are denoted by $i\in\V=\{0,1,\ldots,n\}$ and edges by $\br{i,j}\in\E$. Edges are directed and for any pair of nodes only one edge is present, $\br{i,j}\in\E$. The direction of edges has no physical significance (the flow of gas can be directed either way) and it is introduced for notational convenience. In some cases we will also denote edges by a single index $k$. Then, $k_1$ is the ``head'' of the edge and $k_2$ is the tail. Let $|\E|=m$.

In the steady-state, the gas network is characterized by pressures at every node and the gas flows over each edge. We denote by $\phi_{ij}$ the flow from $i$ to $j$ over edge $\br{i,j}\in\E$. $\phi_{ij}$ is real and it can be positive or negative. We write $i\rightarrow j$ if there $\br{i,j}\in\E$ and $j\rightarrow i$ if $\br{j,i}\in\E$.  We also denote by $\pi_i$ the squared pressure at node $i$.

The GF equations describe steady-state distribution of pressures and flows over the gas network. GF equations assume that the network is balanced, i.e. the total sum of injections and consumptions is zero. To guaranteer the balance one introduces a special node (usually big producer or consumer of gas, called the ``slack node'') which is flexible in its production/consumption. The slack node is denoted by $0$. The remaining nodes are labelled $1,\ldots,n$.  Under these assumptions the balanced GF equations  without compressors are given by \cite{osiadacz1987simulation,96WS,00WS}:
\begin{subequations}
\begin{align}
	\pi_i-\pi_j=\lambda_{ij}\phi_{ij} |\phi_{ij}|,\br{i,j}\in\E\label{eq:GFBasicA}\\
	q_i=\sum_{i\rightarrow j} \phi_{ij} - \sum_{j\rightarrow i} \phi_{ji},i\in\{1,\ldots,n\} \label{eq:GFBasicB}\\	\pi \geq 0
\end{align}	
\end{subequations}\label{eq:GFBasic}
where $\lambda_{ij}$ is a known characteristic of the pipe $\br{i,j}$ (constructed from the pipe length, diameter, friction factor and the media/gas sound velocity).

Define the $n \times m$ matrices $A,B,C$ with entries given by:
\begin{align*}
&	A_{ik} = \begin{cases}
 1 \text{ if } k_1=i \\
 -1 \text{ if } k_2=i \\
 0 \text{ otherwise } 	
 \end{cases},	B_{ik} = \begin{cases}
 1 \text{ if } k_1=i \\
 0 \text{ otherwise } 	
 \end{cases}\\
 & 	C_{ik} = \begin{cases}
 -1 \text{ if } k_2=i \\
 0 \text{ otherwise } 	
 \end{cases}
\end{align*}
Then, Eq.~(\ref{eq:GFBasicB}) can be written as $A\phi=q$.

To solve the GF Eqs.~(\ref{eq:GFBasicA},\ref{eq:GFBasicB}) means finding the squared pressures $\pi$, and gas flows, $\phi$, along the pipes of the network provided globally balanced injections/consumptions at all the nodes and pressure at a node (typically slack bus) are known. Note that the PF Eqs.~(\ref{eq:GFBasicA},\ref{eq:GFBasic}) may have no physical solution if the pressure set at the slack node is too low, i.e. when at least one $\pi_i$ with $i>0$ drops below $0$. To provide an additional pressure boost compressors are introduced. Depending on the local control scheme implemented, compressors may be of a multiplicative or additive type. Multiplicative compressor boosts the pressure locally by constant multiplicative factor in the direction of the flow and additive compressor (less standard) provides an additive boost along the direction of the flow.

We assume that the compressors are placed at some of the edges in the network, to boost the pressure and improve throughput. Compressors are directional. We assume (without loss of generality) that the orientation of the edge coincides with the compressor direction.

In the presence of multiplicative compression, the GF Eqs.~(\ref{eq:GFBasicA},\ref{eq:GFBasicB}) generalize to
\begin{subequations}
\begin{align}
&A\phi=q \\
&\alpha_{ij}\br{\pi_i-r_{ij}\lambda_{ij}\phi_{ij}|\phi_{ij}|} =\\
&\quad \pi_j+\br{1-r_{ij}}\lambda_{ij}\phi_{ij}|\phi_{ij}|\quad \forall \br{i,j}\in\E \\
&\pi \geq 0
\end{align}	
\end{subequations}\label{eq:GFComp}
where $r_{ij}$ is the relative position of the compressor along the edge $\br{i,j}$ and $\alpha_{ij}>1$ is the compression ratio. We assume here that the compressor boosts pressure in the direction from $i$ to $j$.

It is convenient,  for the purpose of further analysis, to restate Eqs.~(\ref{eq:GFComp}) as follows
\begin{subequations}
\begin{align}
&A\phi=q \label{eq:GFa}\\
&\alpha_{ij}\br{\pi_i-r_{ij}\lambda_{ij}\phi_{ij}\pmo_{ij}} =\pi_j+\br{1-r_{ij}}\lambda_{ij}\phi_{ij}\pmo_{ij}\label{eq:GFb}\\
&\pmo_{ij}^2=\phi_{ij}^2\label{eq:GFc}\\
&\pmo_{ij}\geq 0,\pi\geq 0\label{eq:GFd}
\end{align}	
\end{subequations}\label{eq:GF}
where $\pmo_{ij}$ is a newly introduced auxiliary variable.
To see that Eqs.~(\ref{eq:GFd}) are equivalent to \eqref{eq:GFComp} one observes that \eqref{eq:GFd},\eqref{eq:GFc} imply that $\pmo_{ij}=|\phi_{ij}|$.
Eqs.~(\ref{eq:GF}) constitute a set of $n+2m$ equations in $n+2m$ variables: $\pi\in\R^{n}, \phi\in\R^{m},\pmo\in\R^{m}$. We order the variables as $\vall=\br{\pi,\phi,\pmo}$, and then use $\vall_\pi,\vall_\phi,\vall_\pmo$ to refer to the corresponding blocks of $\vall$.
Eqs.~(\ref{eq:GF}) motivates the following definition.
\begin{definition}
The gas flow operator $F:\R^{n+2m}\mapsto 	\R^{n+2m}$ is defined as
\begin{subequations}
\begin{align}
F\br{\vall}=\begin{pmatrix}
A\phi-q \\
\Ma\pi-\Mb*\phi*\psi\\
\phi^2-\psi^2
\end{pmatrix}
\end{align}	
where $\Ma=\tranb{\diagb{\alpha}B-C},\Mb=\br{\alpha*r+(1-r)}*\lambda>0$.
\end{subequations}\label{eq:GFop}
Then the GF \eqref{eq:GF} become simply
\[F\br{\vall}=0\]
\end{definition}

\subsection{Monotone Operators}\label{sec:Mon}

We now review briefly the theory of monotone operators, as it is relevant to the approach developed in the manuscript. For details and proofs of the results quoted here, we refer the reader to the recent survey \cite{boydmonotone}.
A function $H:\R^k\mapsto\R^k$ is said to be a monotone operator over a domain $\C$ if
\[\inner{H\br{x}-H\br{y}}{x-y}\geq 0 \quad \forall x,y\in\C\]
A monotone operator is a generalization of a monotonically increasing function (indeed, if $k=1$, the above condition is equivalent to monotone increase: $x\geq y \implies H\br{x}\geq H\br{y}$).
$H$ is called strictly monotone if
\[\inner{H\br{x}-H\br{y}}{x-y}>0 \quad \forall x,y\in\C,x\neq y.\]
A common example of a monotone operator is the gradient of a differentiable convex function. 
\begin{definition}[Monotone Variational Inequality]
Let $\C\subset \R^k$ be a convex set and $H$ be a monotone operator over $\C$. The variational inequality (VI) problem associated with $H,\C$ is given by
\begin{align}
\text{ Find } x \in \C \text{ such that } \inner{H\br{x}}{y-x}\geq 0\quad \forall y\in\C\label{eq:VI}
\end{align}
\end{definition}

We quote the following results from the literature on variational inequalities relevant for what follows:
\begin{theorem}\label{thm:VI}
If $H$ is a monotone operator over a convex compact domain $\C$ and \eqref{eq:VI}. Then, the solution set of the variational inequality $\opt{X}$ is convex and non-empty. Further, an approximate solution $x_\epsilon$ satisfying
\begin{align}
 \norm{\opt{x}-x_\epsilon}\leq \epsilon \text{ for some } \opt{x}\in\opt{X}\label{eq:VIapprox}
\end{align}
can be found using at most $O\br{\frac{1}{\epsilon}}$ evaluations of $H$ and projections onto $\C$.
 \end{theorem}
\begin{remark}
In this manuscript, we will be interested in finding zeros of operators corresponding to the GF equations with multiplicative compression. It is easy to see that if there exists a point $\opt{x}\in\C$ with $H\br{\opt{x}}=0$, then this is a solution of the variational inequality. Conversely, if all solutions of the variational inequality are such that $H\br{\opt{x}}\neq 0$, then there is no solution of $H\br{x}=0$ with $x\in\C$.
\end{remark}

 \begin{theorem}\label{thm:VIcond}
Suppose $H$ is differentiable. Then $H$ is monotone over $\C$ if and only if
\[\Jac{H}{x}+\tran{\Jac{H}{x}}\succeq 0 \quad \forall x \in \C\] 	
 \end{theorem}

\section{Characterization of Monotonicity Domains of the Gas Flow Operator}\label{sec:MonGF}

As discussed in the earlier Sections, the zeros of the GF operator correspond to solutions to the GF equations. Furthermore, zeros of monotone operators can be found efficiently. Thus, if we can prove that the GF operator is monotone, then, the GF equations can be solved efficiently. It turns out that the GF operator may not be globally monotone, but is monotone over restricted domains. We start by defining a monotonicity domain formally:

\begin{definition}[Monotonicity Domain]
A convex set $S\subset \R^{n+2m}$ is said to be a \emph{monotonicity domain} of the gas flow operator $F$ if there exist invertible matrices $\Pa$ such that the operator $F_{\Pa}\br{x}=\Pa F\br{ x}$ is monotone over the domain $x\in S$ and
\[\pi,\pmo\geq |phi|,A\phi=q \quad \forall \begin{pmatrix}\pi \\ \phi\\ \pmo \end{pmatrix}\in S\]
\end{definition}
The motivation behind this definition is that, as long as we are interested only in finding zeros of $F$, it does not matter where $F$ or $F_{\Pa}$ is monotone, since $F_{\Pa}\br{x}=0 \iff F\br{ x}=0$ so that zeros of $F$ and $F_{\Pa}$ are identical. The second condition simply ensures that the constraints $\pi,\pmo\geq 0,A\phi=q$ is satisfied at all points in $S$, since we require these constraints as part of the gas flow equations.

\begin{theorem}\label{thm:General}
For any invertible matrix $\Pa \in \R^{(n+2m	)\times (n+2m	)}$, the convex set defined by the following constraints is a monotonicity domain of $F$:
\begin{align}
&\pi\geq 0,\pmo\geq |\phi|,A\phi=q \\
& \Sym{\Pa\begin{pmatrix} 0 & A & 0 \\
-\inv{\diagb{\Mb}}\Ma & \diagb{\pmo} & \diagb{\phi}\\
0 & -\diagb{\phi} & \diagb{\pmo}\label{eq:GasLMI}
\end{pmatrix}}\succeq 0
\end{align}
We denote this set by $\C\br{\Pa}$. If there is a solution to the GF equations in $\C\br{\Pa}$, it must be a solution to the following monotone variational inequality:
\[\kappa\in \C\br{\Pa},\inner{F_{\Pa}\br{\kappa}}{x-\kappa}\geq 0 \quad \forall x\in\C\br{\Pa}\]
\end{theorem}
\begin{proof}
See appendix section \ref{sec:App}
\end{proof}
\begin{remark}
The VI stated here may not have a solution since $\C\br{\Pa}$ is not compact. However, for practical gas networks, there are always bounds on pressures $\pi$ and flow magnitudes $|\phi|$. One can add these bounds on $\pi,\phi,\psi$ to the definition of $\C\br{\Pa}$ to get a compact set before solving the VI. Also, even when the GF equations have a solution in $\C\br{\Pa}$, it is possible that the VI has multiple solutions (including the GF equation solution) and some of these may not be solution to the GF equations. Our numerical experience indicates that this is not a practical problem (ie we always find a GF solution), but we note the technical possibility for completeness. Further, in the next section, we resolve this by picking particular choices for $\Pa$ that eliminate this possibility.
\end{remark}

Theorem \ref{thm:General} is a general result that describes a family of monotonicity domains, parameterized by $\Pa$. However, there are two drawbacks of this approach:
\begin{itemize}
\item It is unclear how to choose $\Pa$. Ideally, we would like to choose $\Pa$ so that $\C\br{\Pa}$ is as large as possible and includes all the practically relevant solutions (solutions satisfying typical bounds on flows that one expects will hold in practice).
\item The constraint \eqref{eq:GasLMI} does not have a simple interpretation, so it is difficult to relate the constraints to typical operational constraints imposed on gas flows.
\end{itemize}

In Section \ref{sec:Explicit}, we describe how to choose $\Pa$ so that the LMI based condition \eqref{eq:GasLMI} simplifies to a condition involving only bounds on the flows.
\subsection{Selection of Monotonicity Domain}\label{sec:Explicit}

In this Section, we show that for particular choices of $\Pa$, the LMI condition \eqref{eq:GasLMI} simplifies significantly. Theorem \ref{thm:Tree} shows that if a condition relating the matrices $\Ma$ and $A$ is satisfied, then the GF operator is monotone for $\pmo\geq 0$.  The basic intuition is that since matrix in \eqref{eq:GasLMI} has constant terms $A,M$, these can be cancelled out to get $0$ by appropriate choice of $\Pa$. Similarly, one can cancel out the terms involving $\phi$ as well, so the overall condition reduces to a simple criterion on $\pmo$.

Our first result shows that for the special case of trees and systems with no compression, which are already known to be efficiently solvable \cite{12BNV}\cite{13MFBBCP}, the GF operator is monotone as long as $\pmo\geq 0$. Thus, any GF solution must lie within the monotonicity domain of the GF operator and can be efficiently found. Thus, our approach recovers the previously known results as special cases.

\begin{theorem}[Exactness for Trees and Systems with No Compression]\label{thm:Tree}
Suppose
\begin{align}
\tran{\Ma}\br{\tran{A}\inv{AA^T}A-I}=0\label{eq:alphaCond}
\end{align}
This is guaranteed if one of the following conditions is satisfied:
\begin{itemize}
	\item The network is a tree.
	\item There are no compressors, that is, $\alpha_k=1$ for all $k\in\E$
\end{itemize}
Then, with the choice
\[\Pa=\begin{pmatrix}
	\tran{\Ma}\tran{A}\inv{A\tran{A}} & 0 & 0 \\
	0      & \diagb{b} & 0 \\
	0      & 0     & \diagb{b}
\end{pmatrix}\]
$\Pa F\br{\vall}$ is a monotone operator over the domain $\psi\geq 0$. Let $\beta>0$ be any positive number. 
Define the domain:
\[\C_\beta=\{\vall:\vall_\pi\geq0, \vall_\pmo\geq \beta,A\vall_{\phi}=q\}\]
Let $\opt{\zeta}$ be any solution to the monotone VI:
\[\vall \in \C_{\beta},\inner{\Pa F\br{\vall}}{x-\vall}\geq 0 \quad \forall x\in\C_{\beta}\]
If $|\opt{\zeta}_\phi|=\opt{\zeta}_\pmo$ and there exists $\pi\geq 0$ satisfying \[\Ma\pi=b*\opt{\zeta}_\phi*\opt{\zeta}_\pmo\]
then $\br{\pi,\opt{\zeta}_\phi,\opt{\zeta}_\pmo}$ is a solution to the GF equations. If not, there are no solutions to the GF equations.
\end{theorem}
\eqref{eq:alphaCond} is a very special condition that is only satisfied for specific choice of $\alpha$ for networks with general topologies. Hence the above theorem has limited applicability for general networks. The following theorem gives a way of choosing $\Pa$ such that any solution to the GF with certain upper and lower bounds on $\psi$ can be found efficiently.

\begin{theorem}[Flow bounds for general system]\label{thm:Special}
Let $\Pa^a\in\R^{n\times n},\Pa^b\in\R^{m\times m},\gamma$ be an optimal solution to the following quasi-convex optimization problem
\begin{align}
& \maxi \gamma \\
& \text{ Subject to } \\
&\begin{pmatrix}
\diagb{\eta} & \Pa^b-\diagb{\Pa^c} \\
\tranb{\Pa^b-\diagb{\Pa^c}} & \frac{1}{\gamma}\diagb{\Pa^c}	
\end{pmatrix}\succ 0 \\
&\Pa^a A=\tranb{\inve{\Pa^b\diagb{b}}\Ma} \\
&\Sym{\Pa^b}_{ii}-\eta_i\geq \gamma\br{\sum_{j\neq i}|\Sym{\Pa^b}_{jj}|}
\end{align}
Then, with the choice
\[\Pa=\begin{pmatrix}
	\Pa^a & 0 & 0 \\
	0      & \Pa^b & 0 \\
	0      & 0     & \diagb{\Pa^c}
\end{pmatrix}\]
For any $\beta>0$, $\Pa F\br{\vall}$ is a monotone operator over the domain $\C_{\beta,\gamma}$ given by:
\[\left\{\vall:\vall_\pi\geq 0,|\vall_\phi|\leq \vall_\pmo,,A\vall_{\phi}=q,\beta\leq \vall_\phi \leq \gamma\beta\right\} \]
Let $\opt{\zeta}$ be any solution to the monotone VI:
\begin{align}
\vall \in \C_{\beta,\gamma},\inner{\Pa F\br{\vall}}{x-\vall}\geq 0 \quad \forall x\in\C_{\beta,\gamma}\label{eq:GFVI}	
\end{align}
If $|\opt{\zeta}_\phi|=\opt{\zeta}_\pmo$ and there exists $\pi\geq 0$ satisfying \[\Ma\pi=b*\opt{\zeta}_\phi*\opt{\zeta}_\pmo\]
then $\br{\pi,\opt{\zeta}_\phi,\opt{\zeta}_\pmo}$ is the unique solution to the GF equations in $\C_{\beta,\gamma}$. If not, there are no solutions to the GF equations in $\C_{\beta,\gamma}$. 
\end{theorem}

\section{Numerical Examples}\label{sec:Num}

We test our results on the simplest non-tree network: A Kite network shown in Fig.~\ref{fig:Kite}.

\begin{figure}[htb]
\begin{center}
\includegraphics[width=.9\columnwidth]{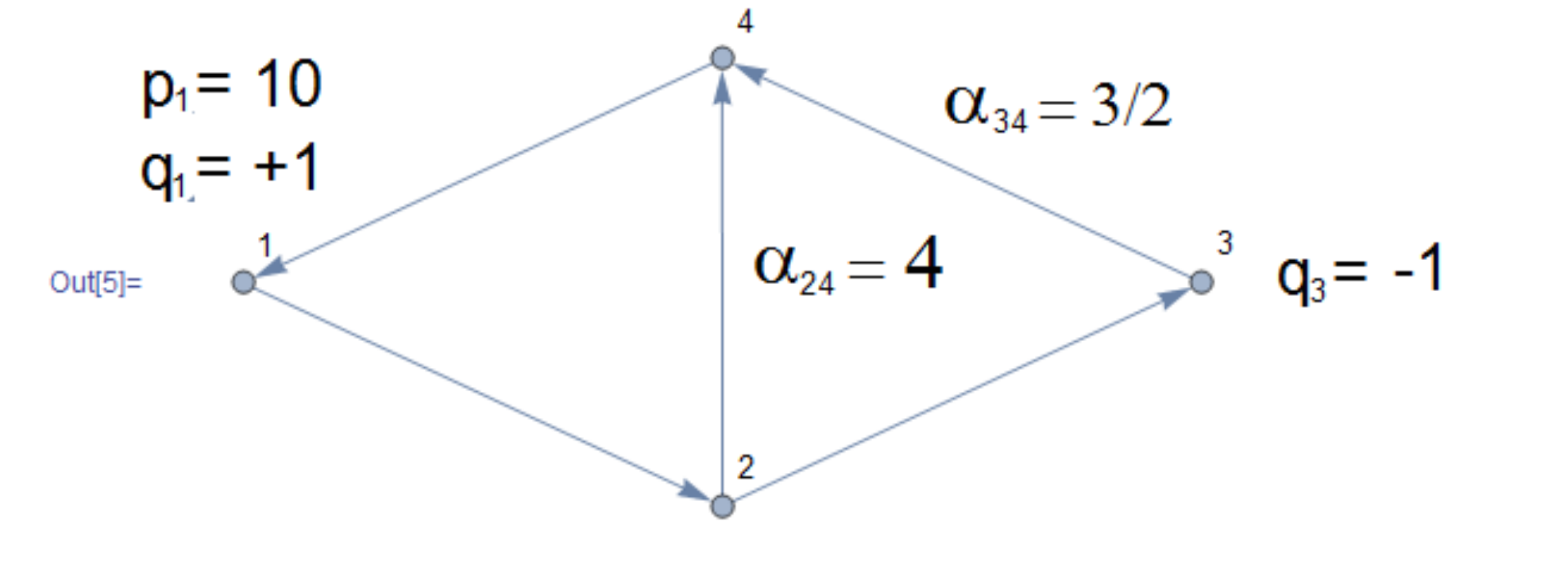}		
\caption{Kite Network}\label{fig:Kite}
\end{center}
\end{figure}

For this network, we find a value of $\gamma=54$.  We observe that for this injection configuration, the flows are well within the bounds imposed by $\gamma$.

\begin{figure}[htb]
\begin{center}
\includegraphics[width=.9\columnwidth]{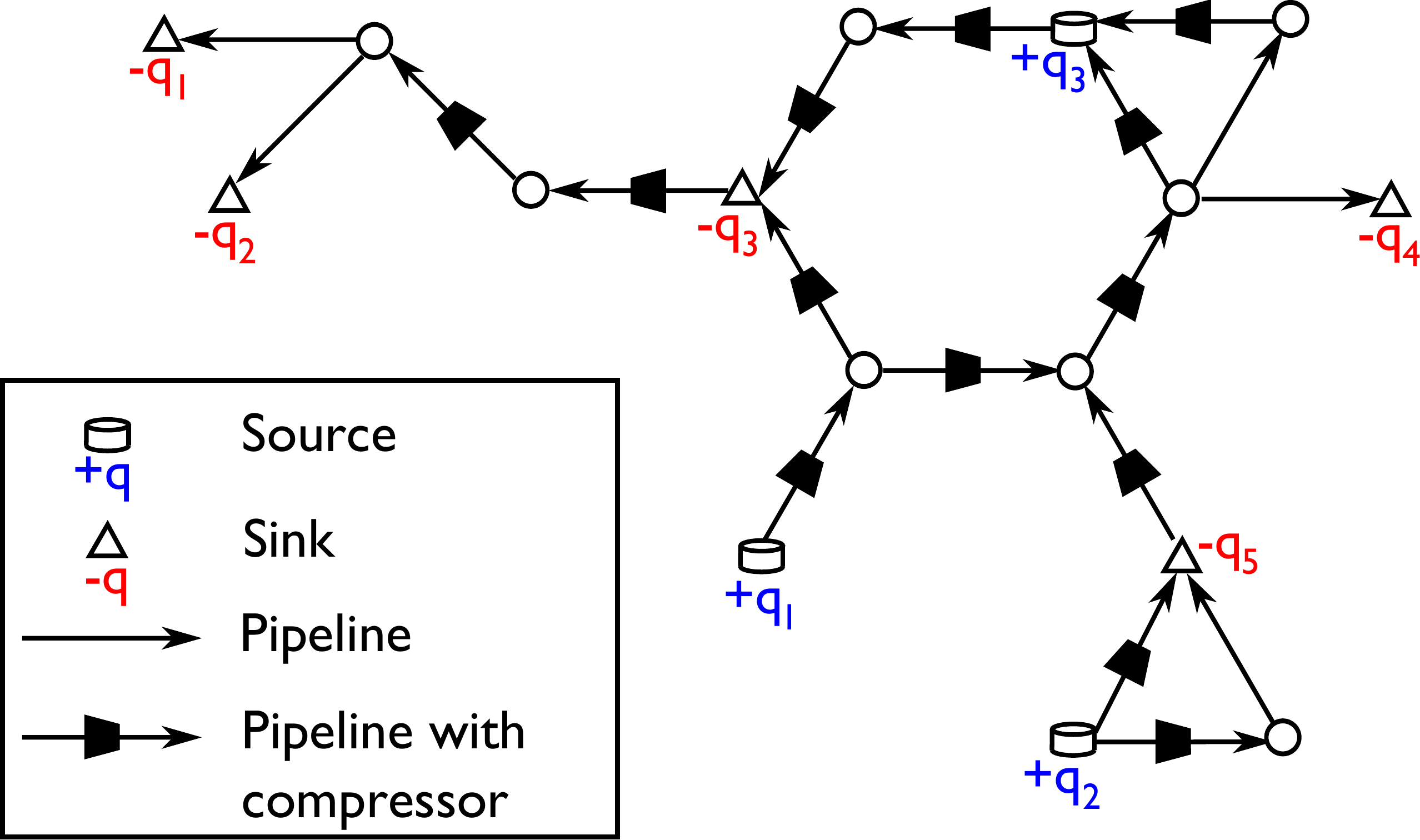}		
\caption{Synthetic 16 bus network}\label{fig:Toy}
\end{center}
\end{figure}

 We also tested our approach on a larger 16 bus network shown in Fig.~\ref{fig:Toy}. We consider three configurations of compressors in the network generated for three different patterns of injection and consumption at the sources and sinks respectively. For these three configurations of compression, we find the values of $\gamma$ to be $470,185,85$ respectively. We observe that all the numbers are well above the actual bounds on the GF. Thus, in all these cases, the GF equations can be solved by solving a monotone variational inequality \eqref{eq:GFVI}.

Numerically, it seems like the value of $\gamma$ reduces as the mean compression in the network increases. To study this effect, we scaled up the compression ratios uniformly and studied the change in $\gamma$. The results are illustrated in Fig.~(\ref{fig:GF}).

\begin{figure}[htb]
\begin{center}
\includegraphics[width=.9\columnwidth]{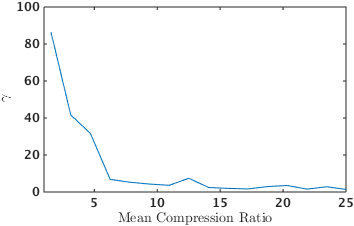}		
\caption{Effect of Compression Ratios on $\gamma$}\label{fig:GF}
\end{center}
\end{figure}

Fig.~(\ref{fig:GF}) shows that the value of $\gamma$ is fairly large (above $10$) for compression ratios of up to about $7$,  which is already well above possible practical value. 
Hence, we believe that, for all practical purposes, our approach can always solve the GF equations (or certify non-existence of practical solutions).

\section{Conclusion}\label{sec:Con}
As the interdependency between gas and power systems increases, it will become critical to have efficient and accurate tools for gas system operations. The steady state GF equations lie at the heart of many computations in the natural gas network  studies related to operations, control, optimization and planning. Therefore, to solve the GF equations in the practical cases of networks, is an important problem which,  to the best of authors' knowledge, was not given much of attention so far.

In this manuscript we remedy the problem and present a novel approach based on monotone operator theory to solve the GF equations. We characterized a whole family of monotonicity domains of the GF equations. Within each of these domains, determining existence of solutions and finding solutions to the GF equations is easy and can be done by solving a monotone variational inequality. Further, we show that by solving a quasi convex optimization problem, one can find a monotonicity domain that contains all GF solutions satisfying certain flow bounds. Numerical studies show that these bounds are sufficiently  tight to capture all solutions of practical interest. Thus, for all practical purposes, our approach can efficiently find (or prove non-existence of) solutions to the GF equations. In future work, we plan  to study if the monotone operator approach developed in the manuscript can be extended to prove that there is always a unique solution to the GF equations (this is indeed what is observed in our experiments).  We plan to focus in the future on developing efficient scalable algorithms to solve large GF problems using the monotone operator approach. Another direction for future analysis is to extend the monotone operator approach to more general dynamic case, of interest in short term (hours and beyond) operations. Finally, we also envision that our approach may be critical for boosting performance of more complex multi-level optimization and control problems where GF are embedded into low-level condition(s).

\section*{Acknowledgment}

The authors thank S. Backhaus for multiple discussions and advice.
The work at LANL was carried out under the auspices of the National Nuclear Security Administration of the U.S. Department of Energy at Los Alamos National Laboratory under Contract No. DE-AC52-06NA25396 and it was partially supported by DTRA Basic Research Project $\#10027-13399$. The authors also acknowledge partial support of the Advanced Grid Modeling Program in the US Department of Energy Office of Electricity.

\bibliographystyle{unsrt}

\bibliography{Ref}

\section{Appendix}\label{sec:App}
\subsection{Proof of theorem \ref{thm:General}}
\begin{proof}
The Jacobian of the gas flow operator $F$ is given by
\begin{align*}
&\begin{pmatrix}
0 & A & 0 \\
\Ma & -\diagb{b*\pmo}	&  -\diagb{b*\phi} \\
0 & \diagb{\phi}	&  -\diagb{\pmo} 
\end{pmatrix}
=	\\
&\begin{pmatrix}
I & 0 & 0 \\
0 & -\diagb{b}	&  0 \\
0 & 0	&  -I
\end{pmatrix}
\begin{pmatrix}
0 & A & 0 \\
-\inve{\diagb{b}}\Ma & \diagb{\pmo}	&  \diagb{\phi} \\
0 & -\diagb{\phi}	&  \diagb{\pmo} 
\end{pmatrix}
\end{align*}
Let $\vall\in\C\br{\Pa}$. It is easy to see that 
\begin{align*}
\Sym{\Pa\begin{pmatrix}
I & 0 & 0 \\
0 & -\inve{\diagb{b}}	&  0 \\
0 & 0	&  -I
\end{pmatrix}\nabla F\br{\vall} }\succeq 0
\end{align*}
Thus, the operator defined by 
\[G\br{x}=\Pa\begin{pmatrix}
I & 0 & 0 \\
0 & -\inve{\diagb{b}}	&  0 \\
0 & 0	&  -I
\end{pmatrix}F\br{ x}\]
satisfies $\Sym{\nabla G\br{x}}\succeq 0$.  Thus $G\br{x}$ is monotone over the domain $x\in\C\br{\Pa}$ and $\C\br{\Pa}$ is a monotonicity domain of $F$.
\end{proof}

\subsection{Proof of theorem \ref{thm:Tree}}
\begin{proof}
With this choice of $\Pa$, the monotonicity domain becomes condition \eqref{eq:GasLMI} becomes
\begin{align*}	
\Sym{\begin{pmatrix} 0 & \Pa^a A & 0 \\
-\Ma & \diagb{b*\pmo} & \diagb{b*\phi}\\
0 & -\diagb{b*\phi} & \diagb{b*\pmo}
\end{pmatrix}}\succeq 0 
\end{align*}	
The matrix evaluates to
\begin{align*}	
{\begin{pmatrix} 0 & \Pa^a A-\tran{\Ma} & 0 \\
\tran{\Pa^a A}-\Ma & 2\diagb{b*\pmo} & 0\\
0 & -0 & 2\diagb{b*\pmo}
\end{pmatrix}}\succeq 0 
\end{align*}	
If $\Pa^a A=\tran{\Ma}$, then this condition is always true when $\pmo\geq 0$. However, $\Pa^a A=\tran{\Ma}$ is an overdetermined equation in general since $\Pa^a\in\R^{n \times n}$ and the number of equations is $nm$ ($\Ma\geq n$ for a connected network). However, if a solution exists, it must also be a solution to 
\[\Pa^a A\tran{A}=\tran{\Ma}\tran{A}\implies \Pa^a = \tran{\Ma}\tran{A}\inv{A\tran{A}}\]
Plugging this back into the original equation, we get
\begin{align*}
&\tran{\Ma}\tran{A}\inv{A\tran{A}}A = \tran{\Ma} \\
&\implies \tran{\Ma}\br{\tran{A}\inv{A\tran{A}}A-I} = 0	
\end{align*}
This is exactly the condition of the theorem.

From the definition of $\Ma$, $\Ma=\tran{A}$ if $\alpha=1$, so that 
\begin{align*}
&\tran{\Ma}\br{\tran{A}\inv{A\tran{A}}A-I}=\\
& A\tran{A}\inv{A\tran{A}}A-A=A-A=0	
\end{align*}
 
If the network is a tree, $A$ is an invertible square matrix, so that 
\[\tran{A}\inv{A\tran{A}}A=\tran{A}\inv{\tran{A}}\inv{A}A=I\]
Hence the condition is satisfied in this case as well.

Uniqueness of solution to the VI: The key observation is that for $\pmo\in\C_\beta$, $F_{\Pa}$ is strictly monotone in its last $2m$ coordinates, since the bottom $2m \times 2m$ sub-block of $\Sym{\nabla F_{\Pa}}$ is positive definite. It is easy to see that if there are two solutions $x,\bar{x}$ to the VI, they must satisfy:
\[\inner{F_{\Pa}\br{x}-F_{\Pa}\br{\bar{x}}}{x-\bar{x}}= 0\]
Now, we can write
\begin{align*}
&F_{\Pa}\br{x}-F_{\Pa}\br{\bar{x}}=\int_{0}^1 \frac{d F\br{\tau x+(1-\tau)\bar{x}}}{d\tau}d\tau \\
& = 	\int_{0}^1 \nabla F\br{\tau x+(1-\tau)\bar{x}} \br{x-\bar{x}}d\tau \\
&\inner{F_{\Pa}\br{x}-F_{\Pa}\br{\bar{x}}}{x-\bar{x}} \\
&=\int_{0}^1 \tranb{x-\bar{x}}\Sym{\nabla F\br{\tau x+(1-\tau)\bar{x}}} \br{x-\bar{x}}d\tau
\end{align*}
Since the bottom $2m \times 2m$ block of $\Sym{\nabla F\br{\tau x+(1-\tau)\bar{x}}}$ is positive definite, the integral is non-zero unless $x_{\phi}=\bar{x}_{\phi},x_{\pmo}=\bar{x}_{\pmo}$. Thus, if there is a solution to the GF equations in $\C_{\beta}$ (which is also a solution to the VI), then given any solution $\opt{\zeta}$ of the VI, we must have $|\opt{\zeta}_\phi|=\opt{\zeta}_\pmo$ and there must exist $\pi\geq 0$ satisfying \[\Ma\pi=b*\opt{\zeta}_\phi*\opt{\zeta}_\pmo\]
Otherwise, there are no solutions to the GF equations in $\C_\beta$.
\end{proof}

\subsection{Proof of theorem \ref{thm:Special}}
\begin{proof}
For this choice $\Pa$, the condition \eqref{eq:GasLMI} 	evaluates to 
\begin{align}	
\Sym{\begin{pmatrix} 0 & \Pa^a A & 0 \\ 
-\Pa^b\inve{\diagb{b}}\Ma & \Pa^b\diagb{\pmo} & \Pa^b\diagb{\phi}\\
0 & -\diagb{\Pa^c*\phi} & \diagb{\Pa^c*\pmo}
\end{pmatrix}} \label{eq:Ref1}
\end{align}
being positive semidefinite. The above matrix evaluates to
\begin{align*}	
\begin{pmatrix} 0 &  Q & 0 \\
\tran{Q} & \Sym{\Pa^b\diagb{\pmo}} & \br{\Pa^b-\diagb{\Pa^c}}\diagb{\phi}\\
0 & \diagb{\phi}\tranb{\Pa^b-\diagb{\Pa^c}} & 2\diagb{\Pa^c*\pmo}
\end{pmatrix}
\end{align*}
where $Q=\Pa^a A-\tran{\Ma}\inve{\diagb{b}}\tran{\Pa^b}$. From the conditions of the theorem, we have $Q=0$. Hence, we only need to prove that
\[\begin{pmatrix}
\Sym{\Pa^b\diagb{\pmo}} & \br{\Pa^b-\diagb{\Pa^c}}\diagb{\phi}\\
 \diagb{\phi}\tranb{\Pa^b-\diagb{\Pa^c}} & 2\diagb{\Pa^c*\pmo}
\end{pmatrix}\succeq 0 \]
Since $\pmo> \beta$, this condition is equivalent to (using Schur complements)
\begin{align*}
&2\Sym{\Pa^b\diagb{\pmo}}\succeq \\
&\br{\Pa^b-\diagb{\Pa^c}}\diagb{\frac{\phi^2}{\Pa^c*\pmo}}\tranb{\Pa^b-\diagb{\Pa^c}}	
\end{align*}
Since $|\phi|\leq \pmo$, $\phi^2\leq \pmo$, so the above condition becomes 
\begin{align*}
&2\Sym{\Pa^b\diagb{\pmo}} \\
&\succeq\br{\Pa^b-\diagb{\Pa^c}}\diagb{\frac{\pmo}{\Pa^c}}\tranb{\Pa^b-\diagb{\Pa^c}}	
\end{align*}
This is true if 
\begin{align*}
& 2\Sym{\Pa^b\diagb{\pmo}}\\
&\succeq\br{\Pa^b-\diagb{\Pa^c}}\diagb{\frac{\beta}{\Pa^c}}\tranb{\Pa^b-\diagb{\Pa^c}}	
\end{align*}

A sufficient condition for the above matrix to be positive definite is
\begin{align*}
&2\diagb{\eta}\beta \succeq \\
&\br{\Pa^b-\diagb{\Pa^c}}\diagb{\frac{\gamma\beta}{\Pa^c}}\tranb{\Pa^b-\diagb{\Pa^c}}	 \\
&\br{\Sym{\Pa^b}}_{ii}-\eta_i\geq \gamma\br{\sum_{j\neq i}|\br{\Sym{\Pa^b}}_{ij}|}
\end{align*}
which is equivalent to
\begin{align*}
&\begin{pmatrix}
2\diagb{\eta} & \br{\Pa^b-\diagb{\Pa^c}} \\
\tranb{\Pa^b-\diagb{\Pa^c}} & \frac{1}{\gamma}\diagb{{\Pa^c}}  
\end{pmatrix}\succeq 0\\
&\br{\Sym{\Pa^b}}_{ii}-\eta_i\geq \gamma\br{\sum_{j\neq i}|\br{\Sym{\Pa^b}}_{ij}|}
\end{align*}
This is exactly the condition that is required of $\Pa$ in the theorem. 

Uniqueness follows from a similar argument as the theorem above.
\end{proof}

\end{document}